\documentclass{IEEEtran}

\usepackage{graphicx,cite,subfigure,bm}

\usepackage{graphicx}
\usepackage{amsmath,amssymb}
\usepackage{ wasysym }
\usepackage{float}
\floatstyle{plaintop}
\restylefloat{table}

\usepackage[ruled,vlined,linesnumbered,noresetcount ]{algorithm2e}

\usepackage{amsmath}

\usepackage{color}
\usepackage{cleveref}
\usepackage{amsthm}
\newtheorem{theorem}{Theorem}[section]

\newtheorem{lemma}[theorem]{Lemma}

 \newtheorem{proposition}[theorem]{Proposition}

 \numberwithin{equation}{section}
\newtheorem{assumption}{Assumption}

 \newcommand{\eps}{\varepsilon}

\begin{document}

\title{\LARGE   Pricing and Routing Mechanisms for Differentiated Services in   \\ an Electric Vehicle  Public Charging Station Network }

\author{Ahmadreza Moradipari, Mahnoosh Alizadeh
\thanks{This work was supported by NSF grant \#1847096. The authors are with the Department of Electrical and Computer Engineering at the University of California, Santa Barbara. \texttt{ahmadreza\_moradipari@ucsb.edu}, \texttt{alizadeh@ucsb.edu}}}

\maketitle
\begin{abstract}
We consider a Charging Network Operator (CNO) that owns a network of Electric Vehicle (EV) public charging stations and  wishes to offer a menu of differentiated service options for access to its stations. This involves designing optimal pricing and routing schemes for the setting where users cannot directly choose which station they use. Instead, they choose their priority level and energy request amount from the differentiated service menu, and then the CNO directly assigns them to a station on their path. This allows higher priority users to experience lower wait times at stations, and allows the CNO to directly manage demand, exerting a higher level of control that can be used to manage the effect of EV on the grid and control station wait times.
We consider the scenarios where the CNO is a social welfare-maximizing or a profit-maximizing entity, and in both cases, design pricing-routing policies that ensure users reveal their true parameters  to the CNO.
\end{abstract}

\section*{Notation} 
 For each customer of type ($i,j,\ell$):
\begin{flalign}
    &\mathcal{V} &&\text{Set of value of times indexed by } v_i \nonumber \\
    &\mathcal{E} &&\text{Set of charging demands indexed by } e_j \nonumber \\
    &\mathcal{B} &&\text{Set of all traveling preferences indexed by } \ell \nonumber \\
    &\mathcal{G}_\ell &&\text{Set of traveling preferences  }\nonumber\\
    &\Lambda_{i,j,\ell} &&\text{Potential expected arrival rate} \nonumber \\
    &R_{i}^\ell && \text{Reward for receiving full battery charge } \nonumber \\
    &q && \text{Charging station } q= 1,\dots,Q \nonumber \\
    &r_{i,j,\ell}^{q} && \text{Routing probabilities to each charging station $q$} \nonumber \\
    &P_{i,j,\ell} && \text{Price of service option } (i,j,\ell) \nonumber \\
    & W_{i,j,\ell}  &&\text{Expected wait time of service option }(i,j,\ell) \nonumber \\
    & d_q  &&\text{Travel time from the main corridor} \nonumber\\ & && \text{to reach station } q \nonumber \\
      & \varrho_q  &&\text{Average queuing time at station } q \nonumber \\
    &\lambda_{i,j,\ell}  &&\text{Effective expected arrival rate }\ell \nonumber \\
    &\boldsymbol{\theta} &&\text{Vector of locational marginal electricity prices   } \nonumber \\
    & C_q  &&\text{Capacity of charging station }q \nonumber \\
    & D  &&\text{Power transfer distribution factor} \nonumber \\
    & E_q  &&\text{Charging station $q$ to load bus mapping matrix} \nonumber \\
    & f_t  &&\text{Line loading limits at each time} \nonumber \\
    &\mathbf{h}_{i,j,\ell}  &&\text{Temporary allocation of admitted customers} \nonumber \\
     & x  &&\text{Lagrange multiplier of capacity constraint} \nonumber 
\end{flalign}
\section{Introduction}\label{sec.intro}
It is well-known that without appropriate demand management schemes in place, Electric
Vehicle (EV) charging patterns could create problems for power transmission and distribution networks, and reduce the environmental benefits of transportation electrification. Hence, the past decade has seen significant research advances in the design of EV demand management algorithms. Broadly speaking, most available smart charging approaches  focus on optimizing  residential and commercial charging profiles when the duration of charge events allows for temporal load shifting.  However, our focus in this paper is on public charging station networks, which are fundamentally different from residential and commercial charging in two ways: 1) Temporal load shifting after a plug-in event is not feasible, unless battery swapping methods are employed. Most  drivers would want to leave the station as soon as possible, quite similar to a gas station stop; 2) Access to EV supply equipment (EVSE) is open to the public, which creates congestion effects and results in wait times at popular stations.

{\it Prior art:} We   categorize the rich literature on mobility-aware charge management of EVs in three categories.  
The first category considers using the mobility pattern of EVs in order to optimize EV charging load in an economic dispatch problem and manage EVs' effects on transmission systems (see, e.g., \cite{khodayar,6729116,sioshansi,8262816,8302928,7967870}) or distribution systems (see, e.g., \cite{6663724,6672275}).  In \cite{7410103}, the authors study the dynamic impact of EV movements on integrated power and traffic systems. They propose Nodal Time-of-Use (NTOU) and Road Traffic Congestion (RTC) prices to control the driving pattern of EV loads.  In \cite{7763884}, the authors study the extended Pickup Delivery Problems (PDPs) for an EV fleet containing EV customers with different service requests. They propose a mixed-integer quadratic constraints optimization for solving the offline pre-trip scheduling problem.  This line of work is not focused on public charging stations and mostly adopts traffic assignment models.
The second category of related work focuses on the problem of routing EV users to stations (see, e.g., \cite{Goeke201581,Wang201537,6805663,pourazarm2016optimal,7001717,5978239}).  Naturally, given the stochastic nature of EV arrivals and limited number of EVSEs at each station, one can consider the problem of managing access to public charging stations as a queuing network, where previous works have considered various objectives such as  revenue maximization or waiting time minimization (see, e.g.   \cite{7042792,qin2011charging, gusrialdi2014scheduling} and the  references therein). The main focus of these papers is the design of optimal {\it routing} policies to directly send users to stations given heterogeneous user needs and not on designing  pricing strategies.   The advantage of using our proposed mechanism compared to these papers is that we jointly design incentive compatible pricing and routing policies. This means that our work does not assume that customers will have to follow our routing orders without considering customers incentives to deviate from the posted assignment. The downside is that our algorithm is more complex than one that is solely focused on optimal routing without any incentive issues. 
The third category of work, which is most intimately connected to this paper, considers the design of pricing strategies to manage users' access to charging networks, where individuals decide which station to use based on prices (self routing) \cite{moradipari2018electric}. In \cite{7835662}, the authors study waiting times of charging station queues and the profit of the CNO under flat rate charging prices as well as a threshold-based pricing policy that penalizes higher demand. In \cite{bayram2013decentralized}, the authors propose a  Stackelberg framework to design prices for charging stations that incentives more uniform station utilization.   In \cite{Liu2019ElectricVE}, the authors study the joint charging and navigation problem of EVs. They formulate en-route charging navigation problem using Dynamic Programming (DP). They propose a so-called Simplified Charge Control (SCC) algorithm for deterministic traffic networks. Moreover, for the stochastic case, they propose an online state recursion algorithm.

Our objective is to guide  EV drivers to drive  \emph{into the right station} in a mobility-aware fashion, in order to 1) manage the effect of EVs on the  grid (e.g., on capacity constrained feeders or integration of behind-the-meter solar) and 2) ensure fair service to customers with proper capacity allocation and short station wait times (admission control), considering heterogeneous user preferences and needs. This is not an easy task to achieve merely through pricing algorithms, mainly due to the complexity of  the price response structure of users and its dependence on the users' mobility needs and preferences, information which is not readily  available and is very hard to obtain.  Hence, we take a different path here, which allows us to somewhat separate the pricing and admission control aspects of the problem. We assume that customers cannot directly choose which charging station they will charge at. Instead, a Charging Network Operator (CNO) is in charge of directly assigning users to charging stations given their respective {\it value of time} (VoT), {\it charging demand} and {\it travel preferences}.  We believe that this is reasonable given that, even today, access to public charging stations is only allowed for specific vehicle types or with users with prepaid charging plans/subscriptions.  A customer's travel preferences specify which charging stations they are willing to visit. The CNO's goal is to design a menu of {\it differentiated service options} with service qualities that are tailored to the characteristics of heterogeneous users. Each service option is tailored to users with given VoT, charging demand, and travel preferences, and is associated with a routing policy (i.e., the probability of that customer type being assigned to each of the stations on their path), as well as an appropriate price. The CNO wishes to optimize these differentiated routing policies and  prices  in order to optimally use capacity-limited charging stations and minimize electricity costs. Furthermore, the CNO's goal is to design {\it incentive-compatible} pricing-routing policies, which ensures that individual users reveal their true needs and preferences to the CNO. Such differentiated pricing mechanisms have been studied before in the context of residential demand response in recent years (see, e.g., \cite{bitar,6876120})   in order to incentivize the participation of loads in direct load control programs, analogous to what we are trying to achieve here for fast charging networks.
In the context of electric transportation systems, in \cite{7914766}, the authors propose differentiated incentive compatible pricing schemes to manage a single charging station in order to increase smart charging opportunities by incentivizing users to have later deadlines for their charging needs (i.e., offer more laxity). Another line of work which inspires the models we adopt in this paper is that of service differentiation techniques in queuing networks, see, e.g., \cite{katta05,Dovrolis,yahalom2006designing,bradford1996pricing}. 

 The contributions of this paper are as follows: 1) Modeling the decision problem faced by a CNO for managing EVs in a public charging station network through differentiated services (Section  \ref{sys.model}); 2) Proposing incentive compatible pricing and routing policies for maximizing the social welfare (Section  \ref{section.so}) or the profit of the CNO (Section \ref{section.pm}) considering users’ mobility patterns, distribution network constraints or behind-the-meter solar generation; 3) Proposing an algorithm that finds the globally optimal solution for the CNO’s non-convex objective in the special case of hard capacity constraints in both social welfare and profit maximization scenarios; 4) Numerical study of the benefits of differentiated services for operation of fast charging networks (Section \ref{section.sim}). A preliminary version of this work is presented in \cite{8619218}. In this work, we add heterogeneous traveling preferences for EV users, we extend our results for a  profit maximizing CNO, and we account for the usage of time-varying behind-the-meter solar energy.

\section{System Model}\label{sys.model}
\subsection{Individual User Model}
We first describe the individual EV users' parameters and decision making model.

\subsubsection{User types}
We assume that users  belong to one of $V\times E \times B$ types. A type $(i,j,\ell)$ customer has a value of time (VoT)  $v_i$ with $i \in \mathcal V = \{1,\ldots, V\}$, an energy demand $e_j$  with $j\in \mathcal E = \{1,\ldots,E\}$, and a traveling preference $\mathcal{G}_{\ell}$, with $\ell \in \mathcal B = \{1,\ldots,B\} $.  The value of time is often used to model the heterogeneity of users’ utility and choice when optimizing their response in the presence of travel time variations. The set of traveling preferences $\mathcal{G}_{\ell}$ declares the set of stations to which  customers with preference $\ell$ have access on their path. More specifically, for each traveling preferences $\ell$, we define the vector ${\bf y}_{\ell}$ with length $Q$ (number of charging stations) such that ${\bf y}_{\ell}( q) = 1$ if station $  q \in \mathcal{G}_{\ell}$ and $0$ otherwise.
For convenience, we order the customer types such that both VoT and energy demand are in ascending order, i.e., $v_1 < v_2 < \ldots < v_V$, and $e_1 < e_2 < \ldots < e_E$. In this paper,
we assume that users do not act strategically in choosing the amount of energy they need, i.e., they fully charge their EV if they enter a charging station. 


  We assume that type $(i,j,\ell)$ customers arrive in the system with a given 3-dimensional expected (average) arrival rate  matrix $\mathbf{\Lambda} = [\Lambda_{i,j,\ell}]_{i \in \mathcal V, j \in\mathcal E, \ell \in \mathcal B}$, which we consider as an inelastic and  known parameter.
   In each {\it potential  arrival}, the customers can choose to either purchase a service option from the differentiated service options offered by CNO, or choose to not buy any charging services.   Note that we are in a static setting, i.e., the expected rate of arrival of users of different types is assumed as a constant variable when designing pricing/routing policies. While the arrival rate can vary across time, we will assume that the dynamics of the charging process at fast charging stations is faster than the dynamics of average traffic conditions.  
   
\subsubsection{Service options}

We assume that the number of differentiated service options that are available matches the three-dimensional user types $(i,j,\ell) \in \mathcal V \times \mathcal E \times \mathcal B$. The CNO will sell each service option $(i,j,\ell)$ with price $P_{i,j,\ell}$. Moreover, service options are differentiated in terms of a routing policy $\mathbf r_{i,j,\ell} = [r^q_{i,j,\ell}]_{q = 1,\ldots,Q}$, which is a column vector of routing probabilities of customers that purchase service option $(i,j,\ell)$ to each charging station $q \in \mathcal G_\ell$.

The joint choice of these {\it pricing-routing policies} $(P_{i,j,\ell}, \mathbf r_{i,j,\ell})$ would affect  the proportion of users that choose to purchase  each service option, which would in turn affect the arrival rate and average charging demand per EV at each charging station. As a result, the average total electricity demand and waiting times at the station are determined through the design of these pricing-routing policies. Hence, the design of the pricing-routing policy to be employed directly affects the social welfare (or the CNO's profit). To concretely model this connection, we first model how users choose which service type to purchase (if any). 

\subsubsection{User decision model}
In general, users have no obligation to buy the services option corresponding to their own true type (why would I tell a CNO that I have low value of time and be assigned a longer wait?). The total utility of a user from purchasing charging services is the reward they receive from charging minus the expected waiting cost (which is the product of VoT with the expected waiting time) and the price paid for charging services.
Let us assume that customers with value of time $v_i$ and traveling preference  $\ell$ will get a reward $R_i^{\ell}$ for receiving full battery charge.  Furthermore, we assume that information about the expected wait time $W_{i,j,\ell}$ of each  service option $(i,j,\ell)$ in the menu is available to users.  Throughout this paper, we assume that the time it takes to drive to a station from the main corridor (denoted by $d_q$) is included in the ``wait time'' corresponding to that station (on top of the queuing time $\varrho_{q}$), i.e., we have 
\begin{align}
    W_{i,j,\ell} = \sum_{q=1}^{Q} \bigg( d_q + \varrho_{q} \bigg) r_{i,j,\ell}^{(q)}.
\end{align}
We will assume that the users do not observe the current exact realization of wait times, i.e., the expected wait time $W_{i,j,\ell}$ is not conditioned on the realization of the random arrival rate of user and will be constant at the equilibrium.
  Therefore, customers of type $(i,j,\ell)$ will choose their service option $(m,k,t)$ by solving:
\begin{equation}\label{cust.benef}
\max_{m \in \mathcal V, j \leq k \leq E, t \in \mathcal  B_{\ell}} R_i^{\ell} - v_i W_{m,k,t} - P_{m,k,t}.
\end{equation}
According to our assumption on the inelasticity of user's charging needs,  customer of type $(i,j,\ell)$ can only choose a service option $(m,k,\ell)$ if $e_j \leq e_k$. Moreover, we assume users of type $(i,j,\ell)$ may only choose a travel preference $t \in \mathcal  B_{\ell}$, where $\mathcal B_{\ell}$ is defined as the set of all preferences $t\in \mathcal B$ such that $\mathcal G_t \subset \mathcal G_{\ell}$ (otherwise the user would have to change their travel origin-destination pair).
If the total utility defined in \eqref{cust.benef} is not positive for any available service option $(m,k,t)$, then that customer will not purchase charging services.   We would like to note that our scheme is not forcing any user to accept the CNO’s routing to different stations. It only provides lower prices for more flexibility in regard to waiting time and station choice. If a user is not willing to provide this flexibility, they may choose to select the service option that only includes the specific station they would like to visit and naturally pay a higher price for receiving service.

The aggregate effect of each individual customer's decision of whether to buy service or not and their choice of service option will lead to a Nash Equilibrium (NE) of {\it effective expected arrival rates} in the charging station network, denoted by $\boldsymbol{\lambda} = [\lambda_{i,j,\ell}]_{i \in \mathcal V, j \in\mathcal E, \ell \in \mathcal B}$.  It is shown in \cite{Schmeidler1973, MASCOLELL1984201} that the Nash equilibrium always exists in the non-atomic game where each user's set of strategies is continuous and measurable.
Our goal in this paper is to design a pricing routing policy such that 1) the resulting NE is optimal for maximizing social welfare or CNO profit; 2) we belong to the family of incentive-compatible (IC) pricing policies, i.e., policies where every user can achieve the best outcome for themselves by acting according to their true preferences. 

Next, we characterize conditions that should hold at equilibrium for such policies.


\subsection{Incentive Compatible (IC) Pricing-Routing Policies}
In this paper, we would like to focus on Incentive Compatible (IC) pricing-routing policies.
A pricing-routing policy is IC if, for each user type $(i,j,\ell)$, it is always optimal to choose the service option that matches their user type, i.e., service option $(i,j,\ell)$. Hence, no users will have any incentive to lie about their user type to the CNO, which can be desirable for system design purposes. 
Mathematically, given the user's decision problem in \eqref{cust.benef}, this condition will be satisfied for a pricing routing policy if  the following conditions are satisfied at equilibrium:
\begin{align}
&\forall k,t \in \mathcal V, t \neq k, \forall j \in \mathcal E, \forall \ell \in \mathcal B \nonumber\\
&~~~~~~~~ P_{k,j,\ell} + v_{k} W_{k,j,\ell} \leq  P_{t,j,\ell} + v_{k} W_{t,j,\ell}, \label{IC1}\\
&~~~~~~~~  P_{t,j,\ell} + v_{t} W_{t,j,\ell} \leq  P_{k,j,\ell} + v_{t} W_{k,j,\ell}, \label{IC2}\\
& \forall i \in \mathcal V, \forall t, k \in \mathcal E, t > k, \forall \ell \in \mathcal B \nonumber\\
&~~~~~~~~  P_{i,k,\ell} + v_{i} W_{i,k,\ell} \leq  P_{i,t,\ell} + v_{i} W_{i,t,\ell},  \label{IC3}\\
&\forall i \in \mathcal V, \forall j \in \mathcal E, \forall  \ell \in \mathcal B, \forall t \in \mathcal B_{\ell} \nonumber\\
&~~~~~~~~  P_{i,j,\ell} + v_{i} W_{i,j,\ell} \leq  P_{i,j,t} + v_{i} W_{i,j,t},  \label{IC4}
\end{align}
These conditions ensure that no user receives a higher utility by joining the system under any type other than their own.
For convenience, we refer to \eqref{IC1}-\eqref{IC2} as {\it vertical IC} constraints,  and   \eqref{IC3} as the {\it horizontal IC} constraint.
Note that while the service options' prices $P_{i,j,\ell}$ play a direct role in these conditions, the routing probabilities ${\bf r}_{i,j,\ell}$ only indirectly affect these conditions by determining the wait times $W_{i,j,\ell}$. We will explore this connection more later. 

Furthermore, Individual Rationality (IR) is satisfied if the following constraints are satisfied at equilibrium:
\begin{align}
P_{i,j,\ell} &= R_i^{\ell} - v_i W_{i,j,\ell}, ~\mbox{if~} 0<\lambda_{i,j,\ell}<\Lambda_{i,j,\ell}\nonumber\\
P_{i,j,\ell} &<  R_i^{\ell} - v_i W_{i,j,\ell}, ~\mbox{if~}  \lambda_{i,j,\ell} = \Lambda_{i,j,\ell} \nonumber\\
P_{i,j,\ell} &>  R_i^{\ell} - v_i W_{i,j,\ell}, ~\mbox{if~} \lambda_{i,j,\ell} =0.\label{ir.const}
\end{align}
That is, for if any user of type $(i,j,\ell)$ joins the system, their utility from joining the system must be non-negative.
Next, we study the structure of NE under any IC policies under two assumptions about rewards $R_i^{\ell}$.
\begin{assumption} For customers with different traveling preference, the  rewards $R_i^{\ell}$ satisfy the following:
\begin{align}
& \forall i \in \mathcal V, \forall \ell,m \in \mathcal B: \nonumber\\&
\text{if} ~ |\mathcal{G}_{\ell}| > |\mathcal{G}_{m}| ~ \text{then} ~ R_i^{\ell} < R_i^{m},  \nonumber\\&
\text{if} ~ |\mathcal{G}_{\ell}| = |\mathcal{G}_{m}| ~ \text{then} ~ R_i^{\ell} = R_i^{m}.
\end{align}
This means that users with a more limited set of charging options get a higher reward from receiving service.
\end{assumption}
\begin{assumption} For customers with the same traveling preference ${\ell}$, the ratios $\frac{R_i^{\ell}}{v_i}$ satisfy the following:
\begin{equation}
\frac{R_1^{\ell}}{v_1} < \frac{R_2^{\ell}}{v_2} < \ldots  < \frac{R_V^{\ell}}{v_V}.
\end{equation}
\end{assumption}
A similar structure was assumed in \cite{katta05} and other past work for service differentiation through pricing-routing policies  in a single  server service facility with Poisson arrivals and exponential service time  M/M/1.

The next lemma shows that under an IC pricing-routing policy, waiting time is a non-increasing function of VoT for users with the same traveling preference and energy demand.
\begin{lemma} \label{icwaitinglemma}
Under an incentive-compatible pricing-routing policy, for any users of types $(i+1,j,\ell)$ and $(i,j,\ell)$ who have purchased charging services, we must have: \begin{equation} W_{i+1,j,\ell} \leq W_{i,j,\ell}.\end{equation}
\end{lemma}
\begin{proof}
From vertical IC constraints \eqref{IC1} and \eqref{IC2}  for customers of type $(i,j,\ell)$ and $(i+1,j,\ell)$, we can write:
$$(v_{i+1} - v_i)W_{i+1,j,\ell} \leq (v_{i+1} - v_i)W_{i,j,\ell},$$
and the fact that  $v_{i+1} - v_i>0$, would lead to the result.
\end{proof}

The next lemma shows that it suffices to only check IC conditions for neighboring service options, e.g., the   options with one level higher value in VoT or energy.
\begin{lemma} \label{localic lemma}
(Local IC) The IC constraints \eqref{IC1}-\eqref{IC4} are satisfied if and only if:
\begin{align}
&\forall i \in \{1,\ldots, V-1\},\forall j \in \mathcal E, \forall \ell \in \mathcal B: \nonumber\\
&~~~~~~~~~ P_{i+1,j,\ell} + v_{i+1} W_{i+1,j,\ell} \leq  P_{i,j,\ell} + v_{i+1} W_{i,j,\ell}, \nonumber\\
&~~~~~~~~~ P_{i,j,\ell} + v_{i} W_{i,j,\ell} \leq  P_{i+1,j,\ell} + v_{i} W_{i+1,j,\ell}, \nonumber\\
&\forall i \in \mathcal V, \forall  j \in \{1,\ldots, E-1\}, \forall \ell \in \mathcal B: \nonumber\\
&~~~~~~~~~ P_{i,j,\ell} + v_{i} W_{i,j,\ell} \leq  P_{i,j+1,\ell} + v_{i} W_{i,j+1,\ell}, \nonumber\\
&\forall i \in \mathcal V, \forall  j \in \mathcal E, \forall \ell \in \mathcal B, \forall t\in \mathcal T_{\ell}: \nonumber\\
&~~~~~~~~~ P_{i,j,k} + v_{i} W_{i,j,k} \leq  P_{i,j,t} + v_{i} W_{i,j,t},   
\label{localIC3}
\end{align}
where $\mathcal T_{\ell}$ denotes the set of all travel preferences $t \in \mathcal B_{\ell}$ such that $|\mathcal G_t| = |\mathcal G_{\ell}| -1 $.
\end{lemma}
\begin{proof}
The proof is trivial by combining consecutive constraints and is omitted for brevity.
\end{proof}
In the following lemma, we highlight a special structure of users' arrival pattern $\boldsymbol{\lambda}$ at equilibrium under an IC policy. 

\begin{figure}
\centering
\includegraphics[width=0.7\linewidth]{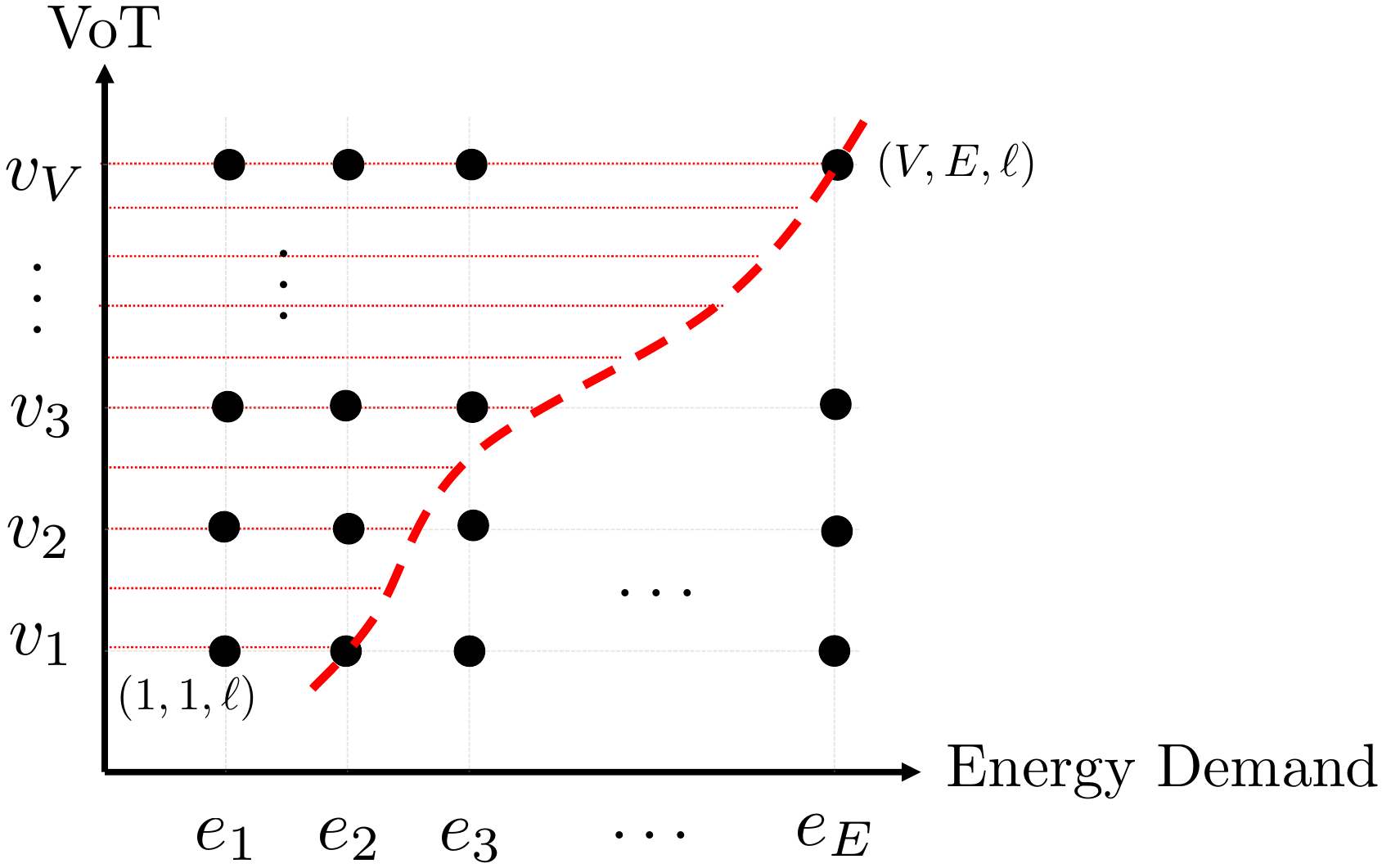}
  \caption{The solution structure for an IC policy.} \vspace{-.2cm} \label{fig:struct}
\end{figure}

\begin{lemma} \label{borderlemma}
If customers of type $(i,j,\ell)$ have  partially entered the system (i.e., $0<\lambda_{i,j,\ell}<\Lambda_{i,j,\ell}$), under an IC policy, the effective arrival rates satisfy:
\begin{enumerate}
\item (Vertical solution structure) $\lambda_{k,j,\ell} = \Lambda_{k,j,\ell}, \forall k>i$, and $\lambda_{k,j,\ell} = 0, \forall k<i$, i.e.,  customers with higher VoTs and similar energy demand and similar traveling preference enter the system in full,  and customers with lower VoTs do not enter the system.
\item (Horizontal solution structure) $\lambda_{i,k,\ell} = \Lambda_{i,k,\ell}, \forall k<j$, and $\lambda_{i,k,\ell} = 0, \forall k>j$, i.e., customers with lower energy demand and same VoT and same traveling preference  enter  in full, and customers with higher energy demand and same VoT and same traveling preference do not enter at all.
\end{enumerate}
\end{lemma}
The proof   follows from combining, IR and IC conditions, as well as  Assumption 1. We omit it due to brevity.


 Therefore, at the Nash equilibrium, due to IC constraints, the solution structure of the effective arrival rates is similar to Fig. \ref{fig:struct}. The red borderline shows which user types should partially enter the system, i.e., where $0<\lambda_{i,j,\ell}<\Lambda_{i,j,\ell}$.  This means that not all users of type $(i,j,\ell)$ will join the system. Hence, from Lemma \ref{borderlemma}, we know that customers to the left of the  line will enter the system in full, and customers to the right  will not enter the system.  Next, we study the design of a socially-optimal IC pricing-routing policy.

%
%

 \section{Socially-optimal policy}\label{section.so}
Our charging stations are located at heterogeneous distances from the users' path and have different locational marginal prices and capacities. In the socially optimal policy, the CNO's goal is to choose a routing policy that maximizes the social welfare of all EV users with access to the network, which we can write as:
\begin{align} 
&\max_{\substack{\mathbf r_{i,j,\ell} \geq 0 \\0\leq \lambda_{i,j,\ell}\leq \Lambda_{i,j,\ell}}} \! \sum_{l = 1}^{B}\sum_{i = 1}^{V}\sum_{j=1}^{E}  \bigg[  R_i^{\ell} \lambda_{i,j,\ell}  - v_i \lambda_{i,j,\ell}
W_{i,j,\ell}(\boldsymbol{\lambda},\boldsymbol{R})  \nonumber \\&  ~~~~~~~~~~~~~~~~~~~~~~~~~~~~~- \boldsymbol{\theta}^T\mathbf r_{i,j,\ell}e_{j}\lambda_{i,j,\ell} \bigg] \label{so.obj}\\
&\mbox{s.t.} ~~~ \mathbf{1}^T \mbox{diag}({\bf y}_{\ell})\mathbf{r}_{i,j,\ell} = 1,~~\forall i\in\mathcal V,j\in \mathcal E,\ell \in \mathcal B,\label{SO.route} \\&~~~\sum_{l = 1}^{B}\sum_{i = 1}^{V}\sum_{j=1}^{E} \lambda_{i,j,\ell} e_j r_{i,j,\ell}^{(q)} \leq C_q, ~\forall q \in \{1,\ldots,Q\}\label{capacity.cons},
\end{align}
where $\boldsymbol{\theta} = [\theta_q]_{q = 1,\ldots,Q}$ denotes the vector of locational marginal prices of electricity at each charging station $q$, $\mathbf{r}_{i,j,\ell} = [r^q_{i,j,\ell}]_{q = 1,\ldots,Q}$ is a column vector of routing probabilities for service option $(i,j,\ell)$ to each charging station $q$, $\mathbf{R} = [\mathbf r_{i,j,\ell}]_{\forall i,j,\ell}$ is the matrix of routing probabilities for all service types, with the $\big[(\ell -1)\times E \times v +  E(i-1)+j\big]$-th column dedicated to type $(i,j,\ell)$, $C_q$ is the capacity of charging station $q$, and $\boldsymbol{\lambda} = [\lambda_{i,j,\ell}]_{\forall i,j,\ell}$ is the vector of effective arrival rates. The objective function is the sum of the reward received by admitted users to the system minus waiting and electricity costs, \eqref{SO.route} ensures that the routing probabilities sum up to one over all charging stations allowed for traveling preference $\ell$, and \eqref{capacity.cons} is the capacity constraint for each charging station. The waiting time function $W_{i,j,\ell}(\boldsymbol{\lambda},\boldsymbol{R})$ maps the effective expected arrival rate in each station into an expected waiting time (e.g., queueing models can be appropriate here).

 Can the CNO design an IC pricing policy which enforces the socially optimal routing solution \eqref{so.obj} as an equilibrium? Next, we propose such a price.
The first order necessary condition for for the problem \eqref{so.obj} is as follows:
\begin{align}
R_i^{\ell} - &v_i W_{i,j,\ell}(\boldsymbol{\lambda},\boldsymbol{R}) - \sum_{t,h,z} \left(\lambda_{t,h,z} v_{t} \frac{\partial W_{t,h,z}(\boldsymbol{\lambda},\boldsymbol{R})}{\partial \lambda_{i,j,\ell}} \right) \nonumber\\& -  \boldsymbol{\theta}^T \mathbf r_{i,j,\ell}e_{j} - {\mathbf{x}}^T \mathbf r_{i,j,\ell}e_j  + \gamma_{i,j,\ell} - \mu_{i,j,\ell}= 0,\label{firstorder}
\end{align}
with $\gamma_{i,j,\ell} \geq  0, \mu_{i,j,\ell} \geq 0$, $\gamma_{i,j,\ell}\lambda_{i,j,\ell} = 0$, $\mu_{i,j,\ell}(\lambda_{i,j,\ell} - \Lambda_{i,j,\ell})=0$, and ${\mathbf{x}} = [x_q]_{q = 1,\ldots,Q}$ as the Lagrange multiplier of the capacity constraint \eqref{capacity.cons}.
We can observe that the following prices will satisfy the IR constraints \eqref{ir.const}:
\begin{align}
P_{i,j,\ell} = \sum_{t = 1}^{V}\sum_{h=1}^{E}\sum_{z = 1}^{B} & \left( \frac{\partial W_{t,h,z}(\boldsymbol{\lambda},\boldsymbol{R})}{\partial \lambda_{i,j,\ell}} \lambda_{t,h,z} v_{t}\right) \nonumber\\& ~ + (\boldsymbol{\theta} + \mathbf{x})^T \mathbf r_{i,j,\ell}e_{j}.\label{marginal}
\end{align}
Next, we  show that the prices in \eqref{marginal}   also satisfy IC constraints \eqref{IC1}-\eqref{IC4}.

\begin{proposition}
With the prices defined in \eqref{marginal}, the solution of socially optimal problem \eqref{so.obj} defines an incentive compatible routing and pricing policy.
\end{proposition}
\begin{proof}
The proof is inspired by that of Theorem 1 in \cite{bradford1996pricing}. To prove incentive compatibility, we need to choose two arbitrary service options and show that with the prices given by \eqref{marginal}, customers from the first type are better off choosing their own option over the other. We first consider vertical IC constraints \eqref{IC1}-\eqref{IC2}. Suppose, we have the globally optimal solution of \eqref{so.obj}. Assume  customers of class $(i,j,\ell)$ enter the system and pretend to be of type $(m,j,\ell)$. We will increase the effective arrival rate of customers of type $(i,j,\ell)$ by an infinitesimal amount  $ \delta$ and treat them as customers of type $(m,j,\ell)$. Hence, because we were at the globally optimal solution of \eqref{so.obj}, we can write:
\begin{align}&\frac{\partial}{\partial { \delta}}\bigg[ R_{i}^{\ell}{ \delta}  - \sum_{(t,h,z) \neq (i,j,\ell)}  v_{t} \lambda_{t,h,z}W_{t,h,z}(\boldsymbol{\lambda} + \boldsymbol{ \delta}_{m,j,\ell},\boldsymbol{R}) \nonumber\\&~ - v_{i} \lambda_{i,j,\ell}W_{i,j,\ell}(\boldsymbol{\lambda}+ \boldsymbol{ \delta}_{m,j,\ell},\boldsymbol{R})  -  { \delta} v_{i} W_{m,j,\ell}(\boldsymbol{\lambda}+ \boldsymbol{ \delta}_{m,j,\ell},\boldsymbol{R})\nonumber\\& -  { \delta} \boldsymbol{\theta}^T\mathbf r_{m,j,\ell}e_{j} - { \delta} {\mathbf{x}}^T \mathbf r_{m,j,\ell}e_j    \bigg]_{{ \delta}=0} \leq 0, \nonumber
\end{align}
Hence, we can write:
\begin{align} R_i^{\ell} &- \sum_{t,h,z} \left( \lambda_{t,h,z} v_{t} \frac{\partial W_{t,h,z}(\boldsymbol{\lambda},\boldsymbol{R})}{\partial \lambda_{m,j,\ell}} \right)  \nonumber\\& -   v_{i}  W_{m,j,\ell}(\boldsymbol{\lambda},\boldsymbol{R})-  \boldsymbol{\theta}^T\mathbf r_{m,j,\ell}e_{j}- { \delta} {\mathbf{x}}^T \mathbf r_{m,j,\ell}e_j    \leq 0.\label{derv.q}\end{align}
Using the price in \eqref{marginal}, this leads to:
$$R_i^{\ell} \leq  v_{i}   W_{m,j,\ell}(\boldsymbol{\lambda},\boldsymbol{R}) + P_{m,j,\ell}$$
and from IR constraints \eqref{ir.const}, we know that if $\lambda_{i,j,\ell}>0$, we need to have $R_i^{\ell}\geq v_{i}   W_{i,j,\ell}(\boldsymbol{\lambda},\boldsymbol{R}) + P_{i,j,\ell}$. Therefore,
$$v_{i}   W_{i,j,\ell}(\boldsymbol{\lambda},\boldsymbol{R}) + P_{i,j,\ell} \leq  v_{i}   W_{m,j,\ell}(\boldsymbol{\lambda},\boldsymbol{R}) + P_{m,j,\ell},$$
which proves that  vertical IC constraints hold.
The proof for \eqref{IC3}-\eqref{IC4} is similar and we remove it due to brevity.
\end{proof}

Our results up to this point are in their most general form. The expected waiting time $ W_{i,j,\ell}(\boldsymbol{\lambda},\boldsymbol{R}) $ associated with each type $(i,j,\ell)$ can be defined using queueing theory as a weighted sum of wait times for the different charging stations, or can have any other general form that arises in reality. However, we would like to note that   the problem \eqref{so.obj} is not convex in general, and hence finding the   solution is not straightforward in all cases. While this is not devastating as this problem only has to be solved for planning,  we will study the  problem in the special case of hard capacity constraints next. This allows us to exploit the special structure  highlighted in Lemma \ref{borderlemma} to  characterize the optimal routing policy through solving linear programs. This is especially useful for our numerical experiments.

{

\subsection{Additional modeling factors: distribution network constraints and behind-the-meter solar }

We would like to note that as opposed to residential and workplace charging, where temporal load shifting is possible for grid support, fast charging stations do not provide such opportunities (unless battery swapping methods are employed).  Our proposed method allows the CNO to consider the following elements when optimization pricing-routing decisions for charging stations: 1) the locational electricity prices for each charging station (already included in \eqref{so.obj}); 2) behind the meter RES supply availability (such as solar generation) at each station; 3) distribution network information and constraints. We will elaborate on the latter two additions in this section.
 
In order to additionally consider network constraints such as line loading limits 
 (defined below as the the total line capacities excluding the loadings induced by conventional demands) the CNO can consider adding the following constraint to the CNO’s optimization problem \eqref{so.obj}:
\begin{align}\label{ntw.const}
    \sum_{q=1}^{Q} D E_{q} \bigg( \sum_{\ell = 1}^{B}\sum_{i = 1}^{V}\sum_{j=1}^{E} \lambda_{i,j,\ell} e_j r_{i,j,\ell}^{(q)} \bigg) \leq f_t, ~\forall t.
\end{align}  The constraint is similar to those that adopted in \cite{6913579,6588984} for temporal load shifting of EV load in distribution networks. The reader should note that  if this constraint is added to \eqref{so.obj},  the Lagrange multiplier of this constraint should be added to the prices we defined in \eqref{marginal}.

Second, we would like to note that behind-the-meter solar energy available at stations can be easily accommodated by our model by adding in virtual stations with electricity price $0$, traveling time equal to the station which is equipped by solar generation, and capacity equal to the currently available solar generation. In this case, the CNO is able to  observe the available behind-the-meter solar integration in real-time, and design pricing-routing schemes in order to efficiently use the real-time solar generation. This addition will help us better highlight the differences between the routing solutions of the social-welfare maximizing and profit maximizing policies that we will discuss in our numerical results in Section \ref{section.sim}. 

}

\subsection{The Special Case of Hard Capacity Constraints}\label{linwait}
In this special case, we assume that station { queuing time (i.e., $\varrho_{ q} = 0$, $\forall { q} = 1,\dots,Q$)}  will be equal to zero as long as the station is operated below capacity. Furthermore, we assume that the travel time from a main corridor to reach each charging station $k$ is a known and constant parameter $d_{ q}, { q} = 1, \ldots, Q$. Therefore, the expected wait time for customers  of type $(i,j,\ell)$ is:
\begin{equation}
   W_{i,j,\ell} = \sum_{{ q}=1}^{Q} d_{ q} r_{i,j,\ell}^{({ q})}. \label{ahmad new watoin}
\end{equation}
Without loss of generality, we assume that stations are ordered such that $d_1 < d_2 < \ldots < d_Q$. We can now rewrite the socially-optimal problem \eqref{SO.route} as:

\begin{align}
\max_{\substack{\mathbf r_{i,j,\ell} \geq 0\\0\leq \lambda_{i,j,\ell}\leq \Lambda_{i,j,\ell}}} \sum_{l=1}^{B} \sum_{i=1}^{V} \sum_{j=1}^{E} \omega_{i,j,\ell},\label{ahmad. raveshe jadide social1}
\end{align}
where
\begin{align}
\omega_{i,j,\ell} = \lambda_{i,j,\ell} \bigg[ & R_i^{\ell} - \bigg( \sum_{{ q}=1}^{Q-1} (v_i(d_{ q} - d_Q) +  e_j(\theta_{q} - \theta_Q))r_{i,j,\ell}^{({ q})}\bigg) \nonumber\\ &~ -(v_id_Q + e_j \theta_Q)\bigg].
\end{align}

 

  We assume that the furthest charging station $Q$ is accessible to all customers with each traveling preference and that $\theta_Q \leq \theta_i, \forall i = 1,\dots,Q-1$. This could represent an inconvenient outside option available to all customers. Additionally, for each charging station $k=1,\dots,Q$, we calculate $ o_s = \big(v_1(d_s - d_Q) +  e_{E}(\theta_s - \theta_Q)\big) $. Then, we label the charging stations with the set ${\bf s} = [s_i]_{i = 1,\dots,Q} $ such that $o_{s_1}\leq o_{s_2}\leq...\leq o_{s_Q}$. The next lemma characterizes the specific order in which customers are assigned to these stations.
 
 \begin{lemma} \label{ahmad structure lemma}
The optimal solution of \eqref{ahmad. raveshe jadide social1} satisfies the following two properties:
\begin{enumerate}
    \item   If customers of type $(i,j,\ell)$ are assigned to  station $s_k$, customers of type $(n,j,\ell)$ with $v_n < v_i$ are only assigned to stations $s_m, m \geq k$. 
    \item  If customers of type $(i,j,\ell)$ are assigned to station $s_k$, customers of type $(i,n,\ell)$ with $e_n>e_j$ are only assigned to stations $s_m, m \geq k$.
\end{enumerate} 
  \end{lemma}
\begin{proof}
We prove both statements by contradiction. Consider the first statement. Suppose there is another optimal solution in which for the customers of type $(n,j,\ell)$ there is a positive probability $r^{(m)}_{n,j,\ell}$ of assignment to station $s_m$ while customers with type $(i,j,\ell)$ have been assigned to a less desirable station $s_k$ with $k>m$. However, we can have another set of routing probabilities such that $r_{n,j,\ell}^{(m)'} = \big(r^{(m)}_{n,j,\ell}-\eps* \lambda_{i,j,\ell}/\lambda_{n,j,\ell}\big)$, $r_{i,j,\ell}^{(m)'} = \eps* \lambda_{i,j,\ell}/\lambda_{n,j,\ell}$, and $r_{i,j,\ell}^{(k)'} = \big(r^{(k)}_{i,j,\ell}- \eps* \lambda_{i,j,\ell}/\lambda_{n,j,\ell}\big)$, which lead to another feasible solution that increases the objective function of \eqref{ahmad. raveshe jadide social1}. Therefore, it is contradictory to the assumption of optimality of the first solution.  The proof of the second statement is similar, and we remove it for brevity.
\end{proof}
  
  \begin{lemma} \label{station.order}
In the optimal solution of problem \eqref{ahmad. raveshe jadide social1}, if charging stations $s_n$ is not used in full capacity, then charging stations $s_m$ with $m>n$ will be empty. 
 \end{lemma}
The proof is provided in the Appendix.


 The takeaway is that in this special case, 1) customers  with higher value of time and lower energy demand receive higher priority in joining stations with lower value of $o_s$; 2) stations are filled in order. This special structure allows us to find the globally optimal solution of non-convex quadratic problem \eqref{ahmad. raveshe jadide social1} by  admitting customers with  higher priority to charging stations with lower value of $o_S$ until they are full. Each station is then associated with a borderline similar to that of Fig. \ref{fig:struct}. User types that fall between the border lines of charging stations $s_{k-1}$ and $s_{k}$ will be routed to charging station $s_k$, whereas user types that fall on the borderline of station $s_k$ will  be partially routed to station $s_k$. User types that fall on the right side of border line of charging station $s_k$ will not be routed to  station $s_k$.

We consider the non-trivial case where all the customers receive positive utility from joining all the charging stations in their traveling preference (otherwise that station will be removed from the preference set). Hence, the CNO will assign customers to charging stations until  either the stations are full or all customers have been admitted. This means that we can assume that the set of available charging stations is:\begin{align}
    \mathcal{X} = \{s_i : v_{V} (d_i - d_Q) + e_1 (\theta_i - \theta_Q) \leq 0 \}, 
\end{align}and the set of potential admittable customers is: \begin{align}
    \mathcal{Y} = \{(i,j,\ell) : R_{i}^{\ell} - \big( v_i d_Q + e_j \theta_Q \big) \geq 0 \}.\end{align}
Exploiting the special solution structure highlighted in Lemmas \ref{ahmad structure lemma} and \ref{station.order}, Algorithm \ref{alg.gloabl}  determines the optimal solution of problem \eqref{ahmad. raveshe jadide social1}. This is done by adding an extra virtual charging station, $s_{Q+1}$, without any capacity constraint such that: \begin{align}
  s_{Q+1} & \in \mathcal{G}_{\ell}, \forall \ell \in \mathcal{B},\\
   \bigg(\max_{\substack{\ell \in \mathcal{B}}} 
    R_{V}^{\ell} \bigg) & < v_1 d_{Q+1} + e_1 \theta_{Q+1}.
\end{align}
Therefore, assigning customers to the charging station $s_{Q+1}$ has negative effect on the social welfare. In step 2, it  admits all types of customers in full, i.e., $\lambda_{i,j,\ell} = \Lambda_{i,j,\ell}, \forall (i,j,\ell)$. After fixing the variable $\lambda_{i,j,\ell} = \Lambda_{i,j,\ell}$, the resulting linear program (LP) of problem \eqref{ahmad. raveshe jadide social1} is referred to as the Border-based Decision Problem (BDP), and its solution determines the temporary allocation (routing probabilities), denoted by $\mathbf{h}_{i,j,\ell}$ = $[h_{i,j,\ell}^{({ q})}]_{{ q} = 1,\dots,Q+1}$, of admitted customers. It  removes the partition of customers that join the virtual charging station as it is shown in step 3.


  \begin{algorithm}\label{alg.gloabl}
 \SetAlgoLined
Add virtual station $s_{Q+1}$ without capacity constraint\\
Set $\lambda_{i,j,\ell} = \Lambda_{i,j,\ell}$, $\mathbf{r}_{i,j,\ell} = \mathbf{0}$ $(\forall i,j,\ell)$ \\
Solve BDP (temporary routing probabilities), and set: \begin{align}
&r_{i,j,\ell}^{({ q})} =  h_{i,j,\ell}^{({ q})} ~\text{for}~ { q}=1,\dots,Q \nonumber\\& \lambda_{i,j,\ell} = \Lambda_{i,j,\ell}(1-h_{i,j,\ell}^{(Q+1)})\nonumber
\end{align} \\

Report the optimal solution : 
 \[
  (\mathbf{R}^{\star}, \boldsymbol{\lambda}^{\star} ) =  \left\{\begin{array}{lr}
       [{r}^{({ q})^\star}_{i,j,\ell}]_{{ q} = 1,\dots,Q} = [r^{({ q})}_{i,j,\ell}]_{{ q} = 1,\dots,Q}\\
        \lambda_{i,j,\ell}^{\star} = \lambda_{i,j,\ell}
        \end{array}\right.
  \]

 \caption{Optimal Admission and Routing}
\end{algorithm}

\begin{theorem}\label{algo.theorem}
 Algorithm 1 will find the globally optimal solution (i.e., the globally optimal effective arrival rates and routing probabilities) for problem \eqref{ahmad. raveshe jadide social1}.
 \end{theorem}
 The proof is provided in the Appendix.
 
Next, we consider the case of designing IC pricing-routing policies for a profit-maximizing CNO.
\section{Profit-maximizing policy}\label{section.pm}
In the section, we study the design of incentive-compatible pricing-routing policies with the goal of maximizing the profit earned by the CNO. Consider the following problem:
\begin{align}
&\max_{\substack{\mathbf r_{i,j,\ell}\geq 0,\\0\leq \lambda_{i,j,\ell}\leq \Lambda_{i,j,\ell}\\ P_{i,j,\ell}}} \sum_{\ell = 1}^{B} \sum_{i = 1}^{V}\sum_{j=1}^{E}  \left[P_{i,j,\ell}\lambda_{i,j,\ell} - \boldsymbol{\theta}^T\mathbf r_{i,j,\ell}e_{j}\lambda_{i,j,\ell} \right].\nonumber\\ &  \mbox{s.t.} ~~~\forall i\in \mathcal V, \forall j \in \mathcal E, \ell\in \mathcal B, \forall m \in \mathcal B_{\ell}:\label{ahmad.profit1} \\
& \sum_{l = 1}^{B}\sum_{i = 1}^{V}\sum_{j=1}^{E} \lambda_{i,j,\ell} e_j r_{i,j,\ell}^{(q)} \leq C_q, ~\forall q \in \{1,\ldots,Q\}, \label{firstprofit.const}\\ & \mathbf{1}^T \mathbf r_{i,j,\ell} = 1, \\
& W_{V,j,\ell} \leq W_{V-1,j,\ell} \leq \ldots \leq W_{1,j,\ell} \leq \frac{R_1^{\ell}}{v_1} , \label{thirdconst}\\
&\sum_{t=1}^{i} (v_{t+1} - v_{t}) (W_{t,j,\ell} - W_{t,j,m}) \leq R_{1}^{m} - R_{1}^{\ell},  \label{new.profit.cons}\\ &\mbox{IC and IR Constraints \eqref{IC1}-\eqref{IC4} and \eqref{ir.const}}.  \nonumber
\end{align}

The CNO's profit is not affected by the average wait times users experience. Instead, the objective function simply considers the revenue from services sold minus the electricity costs. The first and second constraints ensure that station capacity constraints are not violated and routing probabilities sum up to 1. The third (e.g., \ref{thirdconst}) and fourth (e.g., \ref{new.profit.cons}) constraints ensure that the wait times that result from the choice of $\lambda_{i,j,\ell}$ and $\mathbf r_{i,j,\ell}$ do not violate the requirements imposed on wait times in an IC pricing-routing policy. Note that the connection between the prices $P_{i,j,\ell}$ and the admission rate and routing probabilities $\boldsymbol{\lambda}$ and $\boldsymbol{R}$ are only through the IR and IC constraints. Accordingly, for a given set of feasible values of $\boldsymbol{\lambda}$ and $\boldsymbol{R}$, and hence $W_{i,j,\ell}(\boldsymbol{\lambda},\boldsymbol{R})$, one may maximize the prices independently to maximize revenue, as long as IR and IC constraints are not violated.  Consider the following prices:\begin{align}
&\forall j \in \{1,\ldots,E-1\}, \forall i \in \{1,\ldots, V-1\}, \forall \ell \in \{1,\ldots, B\}: \nonumber\\
&P_{i+1,j,\ell} = P_{i,j,\ell} + v_{i+1} W_{i,j,\ell} - v_{i+1} W_{i+1,j,\ell},\label{ahmadic1}\\
&P_{i,j+1,\ell} = P_{i,j,\ell} + v_{i} W_{i,j,\ell} - v_{i} W_{i,j+1,\ell},\label{ahmadic2}\\
&P_{1,1,\ell} =  R_1^{\ell} - v_1 W_{1,1,\ell}.\label{ahmadic3} 
\end{align}
The reader can verify that these prices are as high at horizontal IC constraints allow them to be, and hence, if they are valid, they will be revenue-maximizing. Next, we show that this is indeed the case, i.e., the prices are IC.
 \begin{proposition} \label{pm.ic}
The prices defined in \eqref{ahmadic1}-\eqref{ahmadic3} are Incentive Compatible and Individually Rational. 
\end{proposition}
 The proof is provided in the Appendix.
 
Accordingly, to find the optimal pricing-routing policy, we can simply  substitute  the prices from \eqref{ahmadic1}-\eqref{ahmadic3} in \eqref{ahmad.profit1}, allowing us to  rewrite the problem with fewer decision variables and constraints:
\begin{align}
\max_{\substack{\mathbf r_{i,j,\ell}\\0\leq \lambda_{i,j,\ell}\leq \Lambda_{i,j,\ell}}}&\sum_{l = 1}^{B}\sum_{j = 1}^{E}  \bigg[ \sum_{i=1}^{V} \bigg( R_{1}^{\ell}\lambda_{i,j,\ell} - v_{i}W_{i,j,\ell}(\boldsymbol{\lambda},\boldsymbol{R})\lambda_{i,j,\ell}  -
 \nonumber\\&~~~~~~~~~~~~ \boldsymbol{\theta}^T\mathbf r_{i,j,\ell}e_{j}\lambda_{i,j,\ell}\bigg) - \nonumber \\&~ \sum_{i=1}^{V-1} \bigg((v_{i} - v_{i+1})(\sum_{m =i+1}^{V}\lambda_{m,j})W_{i,1,\ell}(\boldsymbol{\lambda},\boldsymbol{R})\bigg) \bigg].\nonumber\\ \mbox{s.t.} ~~~~~~&  \mbox{Constraints \eqref{firstprofit.const} - \eqref{new.profit.cons}}.\label{ahmad new profit1}
\end{align}
 
 The  profit maximization problem \eqref{ahmad new profit1} has a similar structure to that of \eqref{SO.route}, which we know it is non-convex in general. However, we can still uniquely characterize the globally optimal solution in the special case of hard capacity constraints on charging stations, which is especially helpful in our numerical experiments.

 \subsection{The Special Case of Hard Capacity Constraints}
In the special case of hard capacity constraints, where  \eqref{ahmad new profit1} can be rewritten as: 

\begin{align}
\max_{\substack{\mathbf r_{i,j}\geq 0\\0\leq \lambda_{i,j}\leq \Lambda_{i,j}}}  \sum_{{ q}=1}^{Q-1}\sum_{l=1}^{B}\sum_{i=1}^{V} \sum_{j=1}^{E} \bigg[ R_1^{\ell}\lambda_{i,j,\ell} -  \bigg(\lambda_{i,j,\ell}\big[v_i(d_{ q} - d_Q) \nonumber\\ +   e_j(\theta_{ q} - \theta_Q)\big]r_{i,j,\ell}^{({ q})} +  v_i d_Q + e_j \theta_Q\bigg)~&\nonumber\\   - \left((v_i - v_{i+1})(d_{ q} - d_Q)\left(\sum_{m =i+1}^{V}\lambda_{m,j,\ell}\right)r_{i,1,\ell}^{({ q})}\right) \bigg].&\label{ahmad profite nahaie shode}
\end{align}
We can show that \eqref{ahmad profite nahaie shode} can be similarly solved through BDP linear programs. We remove the details for brevity.

\begin{table}[h]\label{tab:chargingstation}

\centering
\begin{tabular}{ c c c } 
\hline
 Line & L31 & L43\\
 \hline
 limit (kWh) &  7000 & 1400 \\ 
 \hline

\end{tabular}  \caption{Line loading limit}
\end{table}

\begin{table}[h]\label{tab:chargingstation}

\centering
\begin{tabular}{ c|c } 
\hline
 Time travel distance (hour) & Capacity (MWh)\\
 \hline 
 
$d_1 = 0.03 $ &  $c_1 = 0.6$ \\ 
 \hline
$d_2 = 0.06 $  & $c_2 = 0.7$\\ 
 \hline
$d_3 = 0.09$  &   $c_3 = 0.8$ \\ 
 \hline
 $d_4 = 0.12$  & $c_4 = 0.6 $ \\ 
 \hline
  $d_5 = 0.15$    & $c_5 = 0.8 $ \\ 
 \hline
  $d_6 = 0.18$   & $c_6 = 1 $ \\ 
 \hline
\end{tabular}  \caption{Charging stations' values}
\end{table}

 \section{Numerical Results}\label{section.sim}

\subsection{Grid Structure}
{ 
To study the effect of distribution system constraints on the pricing/routing solutions, we use bus 4 of the Roy Billinton Test System (RBTS) \cite{76730}.   Fig. \ref{fig:bus4} shows the single line diagram of Bus 4 distribution networks. Line limit details are shown in Table I.  In the case study, we include $6$ charging stations with parameters shown in Table II. The first three stations are load points LP6, LP7 and LP15 in  bus $2$ of RBTS, and the rest of charging stations are in  bus $4$ of RBTS as shown in Fig. \ref{fig:bus4}.   We assume that each load point with a charging station also has a commercial conventional loading with an average  of 415  kW and a peak of 671.4  KW. 
Furthermore, for each bus, we use the locational marginal electricity prices data from \cite{margin.elec.pric}.  

\begin{figure}[h]
\centering
\includegraphics[width=0.83\linewidth]{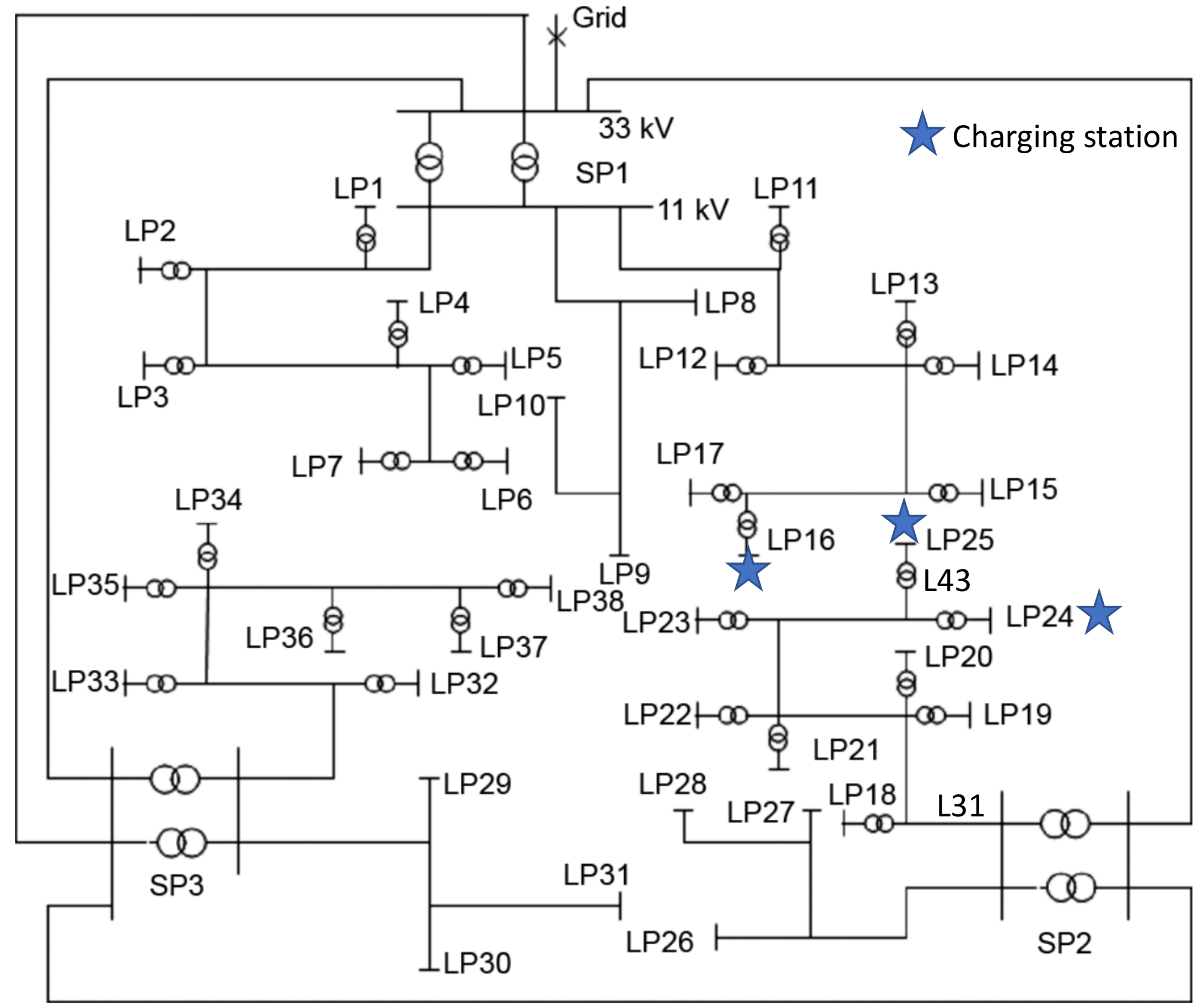}
  \caption{ {  Single line diagram of bus 4 distribution system of RBTS}} \vspace{-.2cm} \label{fig:bus4}
\end{figure}

\subsection{EV Arrival}

In our case study, we assume each customer belongs to one of $125$ user types considering $5$ different value of times, and $5$ different energy demand and $5$ different traveling preferences as it is shown in Table III. We note that the dimension of the type grid is not a major issue and it can be further expanded if needed. We consider 24 time slots with varying potential arrival rates for each day (note that at each time slot, we solve a static problem as we have assumed that the dynamics of charging, which takes around 20 minutes, is faster than the dynamics of the variations of arrival rates). We use the Danish driving pattern in \cite{5751581} to model  EVs arrival rates (see Fig. \ref{fig:numberev}). 

\begin{table}[h]\label{tab:peopletype}

\centering
\resizebox{\columnwidth}{!}{
\begin{tabular}{ c|c|c } 
 \hline
 Value of Time (\$/h) & Energy Demand (kWh) & Traveling Preferences\\
 \hline 
$v_1 = 20  $ &  $e_1 = 30 $   &  $b_1 = \{s_1,s_2\}$  \\ 
 \hline
$v_2 = 30  $ & $e_2 = 40 $ & $b_2 = \{s_3,s_4\}$ \\ 
 \hline
$v_3 = 40  $ & $e_3 = 50 $ & $b_3 = \{s_5,s_6\}$ \\ 
 \hline
 $v_4 = 50  $ & $e_4 = 60 $ & $b_4 = \{s_2,s_3\}$ \\ 
 \hline
 $v_5 = 60  $ & $e_5 = 70 $ & $b_5 = \{s_4,s_5\}$ \\ 
 \hline
\end{tabular}} \caption{{ Customers' types}}
\end{table}

\begin{figure}[h]
\centering
\includegraphics[width=0.9\linewidth]{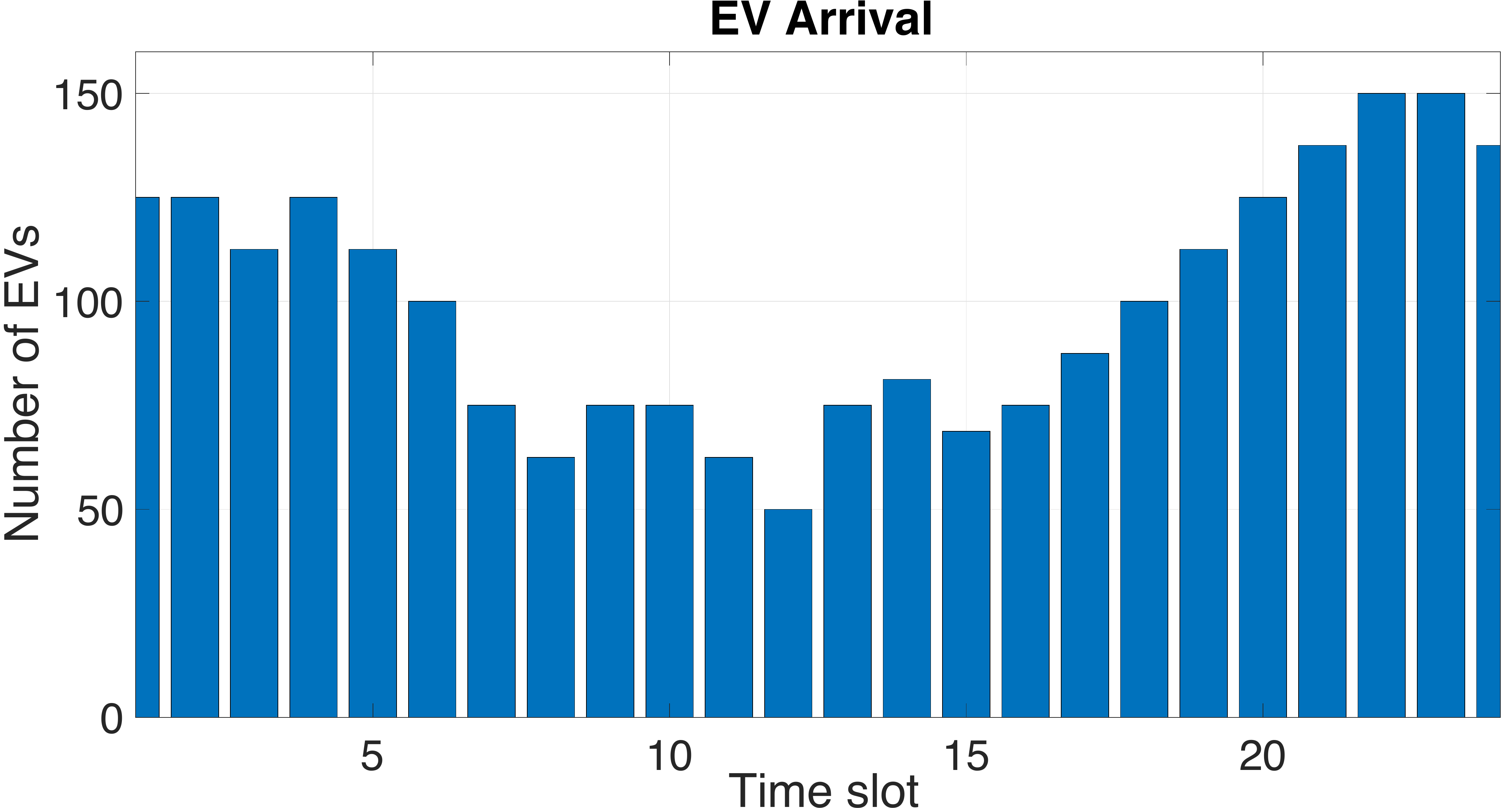}
  \caption{ { EVs arrival to the system at each time step.}} \vspace{-.2cm} \label{fig:numberev}
\end{figure}}

We focus specifically on the special case of stations with hard capacity constraints, where our proposed Algorithm 1 can determine the globally optimal pricing-routing policy. Then we study both socially optimal and profit maximizing scenarios.
We highlight the results of our algorithm by considering both charging stations equipped with behind-the-meter solar generation and without any solar generation.  

{\subsection{ Experiment Results}

In a socially optimal scenario, it can be seen from Fig. \ref{fig:limit}  that line loadings reach but not exceed the limit at hours $14$, $23$ and $24$, which means the distribution network constraints are active for station $6$. Hence, the CNO can design an incentive compatible pricing and routing scheme while considering the impact of EV charging in the power distribution system (in Fig. \ref{fig:limit}, it is shown that in the absence of distribution system constraints, the optimal pricing/routing strategy would violate network constraints).}

\begin{figure}[h]
\centering
\includegraphics[width=1\linewidth]{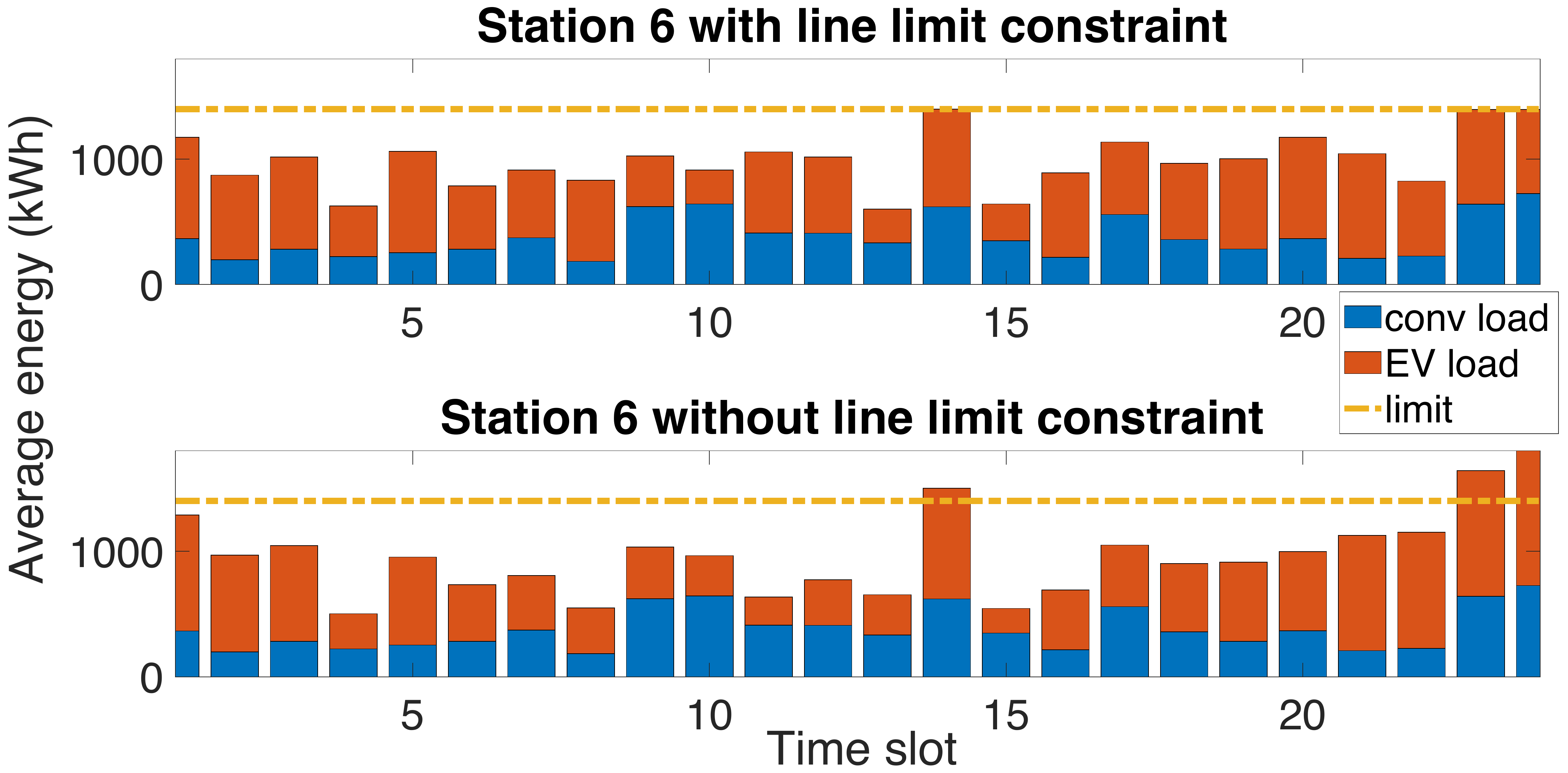}
  \caption{ { Line loading of the socially optimal problem for station $6$.}} \vspace{-.2cm} \label{fig:limit}
\end{figure}

{
Now, let us assume that charging station $6$, which is the farthest charging station from customers routes (i.e., the least desirable assignment for them in terms of traveling distance), can potentially be equipped with a  behind-the-meter  large-scale ($500$kW)   solar system (this will require  $1500$m$^2$ of roof space to install). For the random generation profiles, we use solar data from \cite{solar.eng.jun.2019} for June 2019 (one realization shown in Fig. \ref{fig:social}). 

\begin{figure}
\centering
\includegraphics[width=1\linewidth]{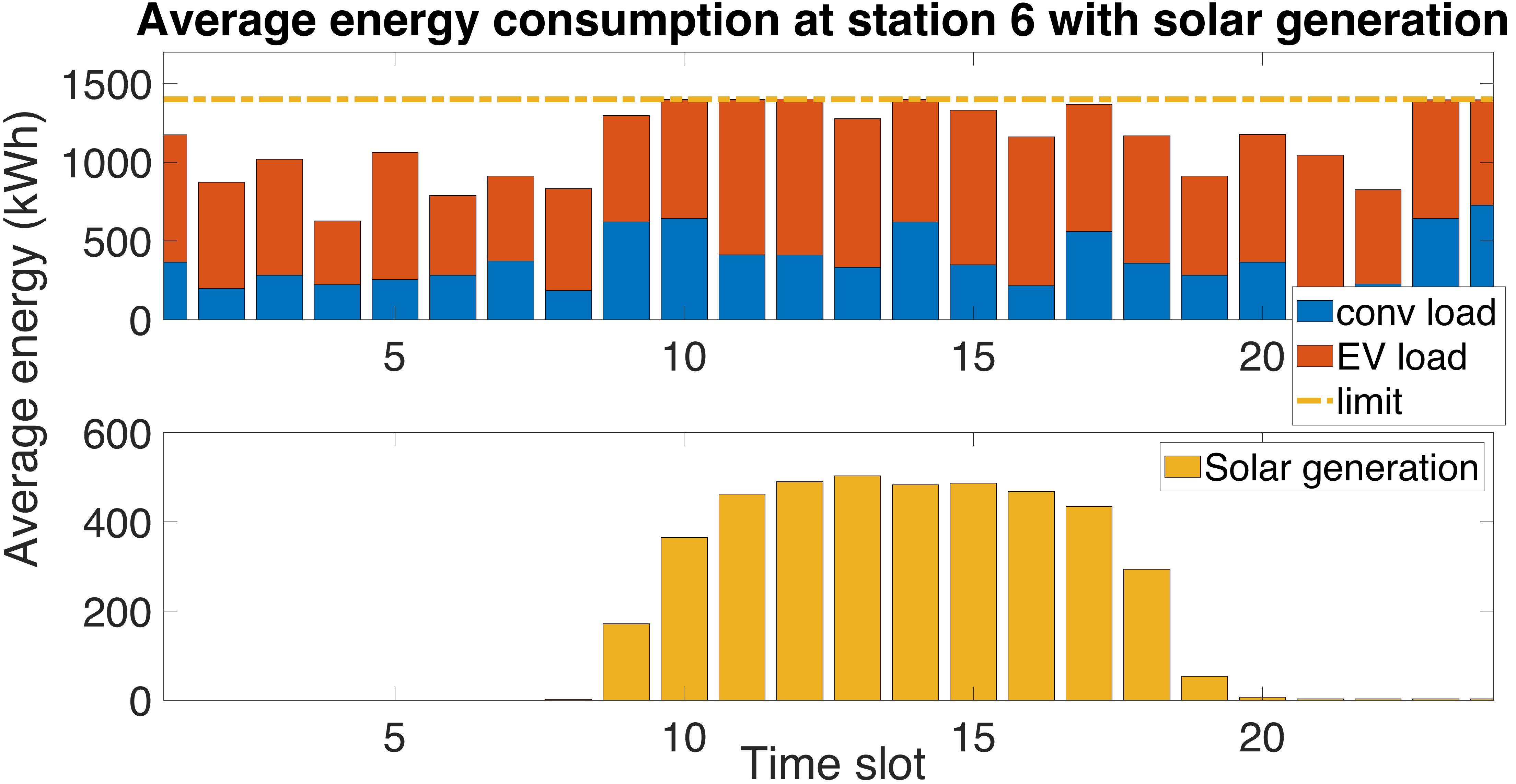}
  \caption{ { Energy demand for charging station 6 with behind-the-meter solar generation capacity.}} \vspace{-.2cm} \label{fig:social}
\end{figure}

The first result we highlight is the energy consumption profile of station $6$ under the social-welfare maximizing scenario with available solar capacity. Essentially, by comparing energy demand with no solar generation, i.e., Fig. \ref{fig:limit} and with solar generation, i.e., Fig. \ref{fig:social}, we see that the availability of free solar energy makes the farthest charging station have higher levels of demand in order to maximize welfare, and so customers have to drive further on average. We will highlight this trade-off more thoroughly next.

 Specifically, Table IV shows the cost of traveling  from the main corridor to reach charging stations over all types of customers with vehicle arrivals shown in Fig. \ref{fig:numberev}. We calculate $\sum_{l = 1}^{B}\sum_{i = 1}^{V}\sum_{j=1}^{E} v_i \lambda_{i,j,\ell}W_{i,j,\ell}$ as the cost of traveling  in both socially optimal and profit maximizing scenarios over a day. Without solar generation, for  both the cases in which the objective is to maximize social welfare and to maximize profits,  customers with a higher VoT and lower energy demand have priority in joining the closer charging stations. With solar generation, in the socially optimal case, customers with higher energy demand are assigned to the furthest charging station even to get cheaper electricity, and the traveling cost is larger. However, for the profit maximizing case, customers with a higher value of time (and hence higher willingness to pay) are still assigned to the closer charging stations (and are charged more), and the overall cost of traveling is less than when the objective is to maximize social welfare, and larger than not having solar generation.
 
 \begin{table}[h]\label{tab:convload}

\centering
\resizebox{\columnwidth}{!}{
\begin{tabular}{ c|c|c } 
 \hline
  &  Socially optimal & Profit maximizing\\
With solar generation  & 9460 (\$) & 9320 (\$)  \\ 
Without solar generation  & 8280 (\$) & 8440 (\$) \\ 
 \hline
\end{tabular}} \caption{{  Cost of traveling of all customers over a day  }}
\end{table}
We would like to note that the concept of incentive-compatibility as highlighted in our paper only applies to each individual’s incentive for incorrectly reporting their type to the CNO under the differentiated service program. The algorithm provides no guarantee that every individual is better off under the differentiated SO policy than they would be under a Nash Equilibrium with no centralized routing, hence incentivizing them to request the existence of the differentiated service program. This is considered normal since any type of congestion pricing mechanism (including locational marginal pricing) to maximize welfare could lead to cost increases for some individuals but overall improve welfare for the society.
}

{ 
\subsection{Bench-marking with status-quo}
The goal of this experiment is to highlight the benefits of a mobility-aware differentiated service mechanism as opposed to self-routing by customers to stations, which can be  considered the status-quo. We have compared the performance of our proposed solution to the equilibrium load and wait time pattern at the stations in the scenario where users self-route. We  assume that in the self-routing scenario, customers will be charged at locational marginal prices for energy (which can vary across stations). For the experiment, we assume 3 different user types, and 3 charging stations (this is clearly not a realistic choice of the parameters, but computing all the equilibria is computationally challenging in bigger cases). The values we used for the numerical experiment are shown in Table V:

 \begin{table}[h]\label{tab:convload}

\centering
\resizebox{\columnwidth}{!}{
\begin{tabular}{ |c|c|c|c| } 
 \hline
  Energy demand (kWh) &  50  & 60 & 40\\
  \hline 
Value of time (\$/h)  & 20 & 30 & 40  \\ 
  \hline
Reward (\$)  & 440 & 635 & 845 \\ 
 \hline
 Locational marginal price (\$/kWh)  & 0.5 & 0.4 & 0.3 \\ 
 \hline
  Time travel distance (h)  & 0.3 & 0.6 & 0.9\\ 
  \hline
\end{tabular}} \caption{{  Parameters }}
\end{table}

Then, we let the customers to selfishly choose the charging station they want to charge at in order to maximize their utility. We need to note that multiple Nash equilibria may exist for this game.{ In our setup, there exist $4$ different equilibria, and the values of social welfare are $ 7290.9 \$, 7302.1 \$, 7312.1 \$, 7328.1\$ $. Observe that they are all less than the value of social welfare achieved using our proposed solution based on differentiated services, which is $7398.9 \$$}. We can argue that this is a natural observation given the lack of appropriate congestion pricing schemes that can deter users from the most popular choice of stations. We note that congestion pricing to guide users towards a socially-optimal charge footprint while considering station capacities is not straightforward to apply in this case for reasons explained in the Introduction. 

}
\section{{ Conclusions and future work}}
We studied the decision problem of a CNO for managing EVs in a public charging station network through differentiated services.  In this case, EV users cannot directly choose which charging station they will charge at. Instead, they choose their energy demand and their priority level, as well as their traveling preferences (which stations they are willing to visit) from among a menu of service options that is offered to them, and the CNO then assigns them to the charging stations directly to control station wait times and electricity costs. This is reminiscent of incentive-based direct load control algorithms that are very popular in demand response. We propose  incentive compatible pricing and routing policies for maximizing the social welfare or the profit of the CNO. We proposed an algorithm that finds the globally optimal solution for the non-convex optimizations that appear in our paper in the special case of hard capacity constraints in both social welfare and profit maximization scenarios and highlighted the benefits of our algorithms towards behind-the-meter solar integration at the station level. {  For future work, we can consider the heterogeneity of customers in assigning different values to different charging stations that have to do with more than just the travel distance to the station and the waiting time in the queue. For example, users might be interesting in accessing some of available shopping options and amenities at particular stations while their vehicle is being charged.
}

\bibliographystyle{IEEEtran}
\bibliography{biblio.bib}

\begin{IEEEbiography}[{\includegraphics[width=1in,height=1.21in,clip,keepaspectratio]{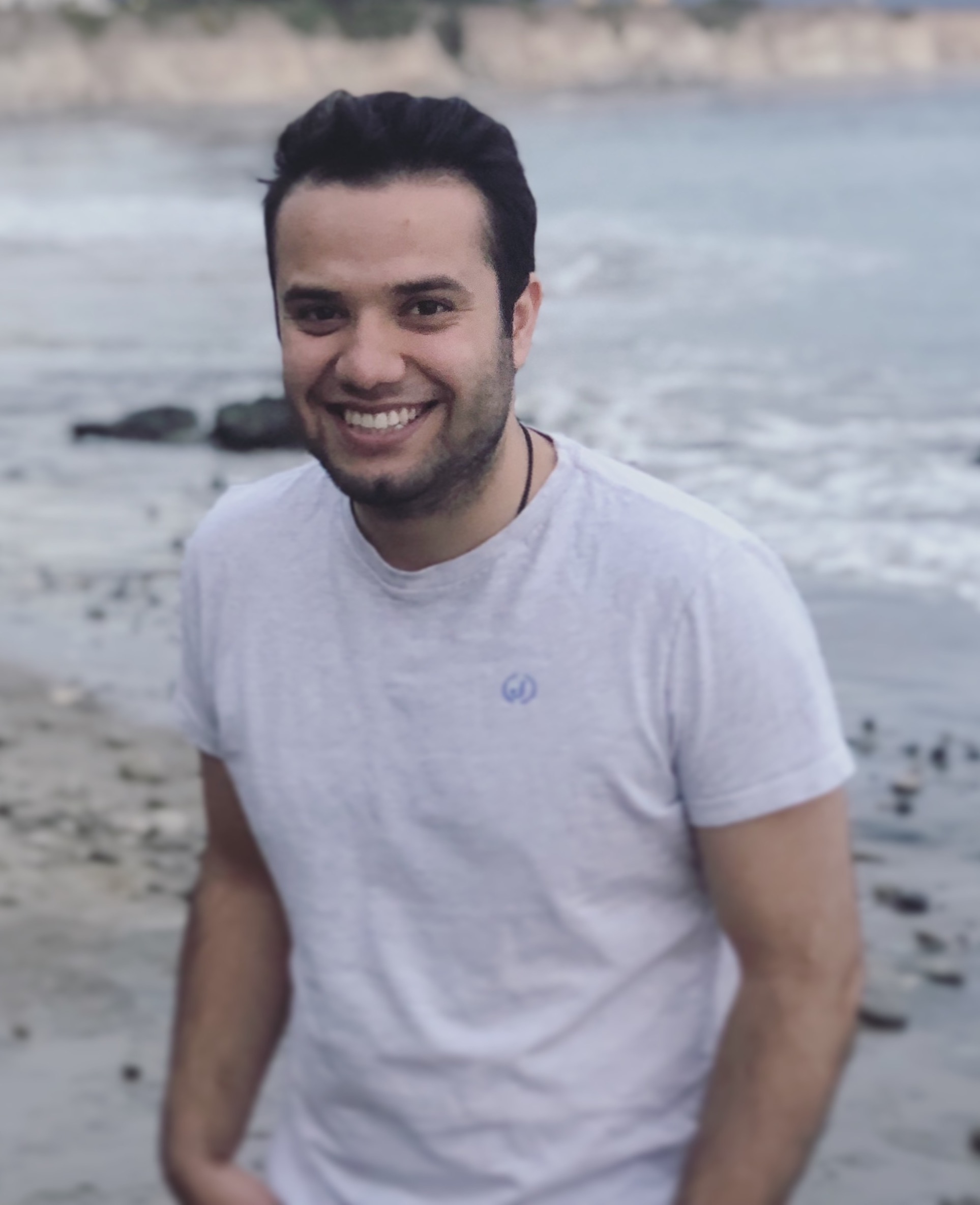}}]{AHMADREZA  MORADIPARI}
  is pursuing the Ph.D. degree at the University of California Santa Barbara. He received the B.Sc. degree in Electrical Engineering from Sharif University of Technology in 2017. His research is mainly focused on optimization and control algorithms to manage congestion and reduce electricity costs in electric transportation systems.
\end{IEEEbiography}
\begin{IEEEbiography}[{\includegraphics[width=1in,height=1.21in,clip,keepaspectratio]{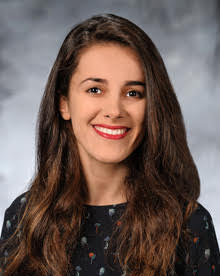}}]{MAHNOOSH ALIZADEH}is an assistant professor of Electrical and Computer Engineering at the University of California Santa Barbara. Dr. Alizadeh received the B.Sc. degree in Electrical Engineering  from Sharif University of Technology in 2009 and the M.Sc. and Ph.D. degrees from the University of California Davis in 2013 and 2014 respectively, both in Electrical and Computer Engineering. From 2014 to 2016, she was a postdoctoral scholar at Stanford University. Her research interests are focused on designing scalable control and market mechanisms for enabling sustainability and resiliency in societal infrastructures, with a particular focus on demand response and electric transportation systems. Dr. Alizadeh is a recipient of the NSF CAREER award.
\end{IEEEbiography}

 \newpage
\appendix


\subsection{Proof of Lemma III.3}
\noindent
 We prove this by contradiction. Suppose there is another optimal solution for problem \eqref{ahmad. raveshe jadide social1} ($\mathbf{R}^{\star}, \boldsymbol{\lambda}^{\star}$) in which for customers with traveling preference $\mathcal{G}_{\ell}$, station $s_n$ has empty capacity, and assume that customers with type $(i,j,\ell)$ are assigned to stations $s_m$ with $m>n$. However, we can have another set of routing probabilities such that for a small $0\leq \epsilon$, $r_{i,j,\ell}^{(n)^{'}} = r_{i,j,\ell}^{(n)^{\star}} + \epsilon $, and $r_{i,j,\ell}^{(m)^{'}} = r_{i,j,\ell}^{(m)^{\star}} - \epsilon$ which is a feasible solution, and it will increase the objective function \eqref{ahmad. raveshe jadide social1} due to the structure we found in lemma \ref{ahmad structure lemma}. Hence, it is contradictory to the optimality of this solution.

 \subsection{Proof of Theorem \ref{algo.theorem}}
 

We first assume that all the charging stations are used in full capacity, i.e., potential   customers are more than the available capacity of charging stations.  We need to show that the Algorithm \ref{alg.gloabl}  will find the optimal solution of problem \eqref{ahmad. raveshe jadide social1}. For convenience, denote as $f(.)$  the objective function of \eqref{ahmad. raveshe jadide social1}, and $g(.)$ as the resulting linear program of problem \eqref{ahmad. raveshe jadide social1} when we consider virtual station $s_{Q+1}$ and we fix $\lambda_{i,j,\ell} = \Lambda_{i,j,\ell}$, $\forall (i,j,\ell)$. Assume the optimal solution of   problem $f$ to be  $A^\star = \bigg( [r_{i,j,\ell}^{(k)^{\star}}]_{k=1,\dots,\rho}, \lambda_{i,j,\ell}^{\star} \bigg)$, $\forall (i,j,\ell)$, and the optimal solution of linear program $g$ to be $\hat{B}^{{\star}} = \bigg([h_{i,j,\ell}^{(k)^{\star}}]_{k=1,\dots,\rho,Q+1}, \lambda_{i,j,\ell} = \Lambda_{i,j,\ell} \bigg)$, $\forall (i,j,\ell)$ . We define $\hat{A} = \bigg([r_{i,j,\ell}^{(k)^{'}}]_{k=1,\dots,\rho,Q+1}, \lambda_{i,j,\ell}^{'} = \Lambda_{i,j,\ell} \bigg)$, $\forall (i,j,\ell)$, such that: \begin{align}
    [r_{i,j,\ell}^{(k)^{'}}]_{k=1,\dots,\rho} &= [r_{i,j,\ell}^{(k)^{\star}}]_{k=1,\dots,\rho},\\
    r_{i,j,\ell}^{(Q+1)^{'}} &= \frac{\Lambda_{i,j,\ell} - \lambda_{i,j,\ell}^{\star}}{\Lambda_{i,j,\ell}},
\end{align}
and we define $B = \bigg([h_{i,j,\ell}^{(k)^{''}}]_{k=1,\dots,\rho}, \lambda_{i,j,\ell}^{''} \bigg)$, $\forall (i,j,\ell)$, such that:\begin{align}
    [h_{i,j,\ell}^{(k)^{''}}]_{k=1,\dots,\rho} &=  [h_{i,j,\ell}^{(k)^{\star}}]_{k=1,\dots,\rho},\\
    \lambda_{i,j,\ell}^{''}  &= \Lambda_{i,j,\ell} (1- h_{i,j,\ell} ^{(Q+1)^{\star}}).
\end{align}
Therefore, $\hat{A}$ and $B$ are in the feasible set of solutions of problems $g$ and $f$, respectively.
\begin{table}[h!]
\centering
\begin{tabular}{ c|c } 
 
$f$ & $g$\\
 \hline 
 $A^{\star}$ & $\hat{A}$\\
 
 $B$ & $\hat{B}^{\star}$

\end{tabular} 
\end{table}
By the definition of optimality, we can write:\begin{align}
    f(A^{\star}) + g(\hat{B}^{\star }) &\geq g(\hat{A}) + f(B), \text{or} \\
   \alpha = f(A^{\star}) - g(\hat{A}) &\geq   f(B) - g(\hat{B}^{\star }) = \beta, \label{ekhtelaf}
\end{align}
where $\alpha$ is the negative effect of admitting all customers to the system and adding the virtual station $s_{Q+1}$ on the problem \ref{ahmad. raveshe jadide social1} for optimal solution $A^{\star}$, and $\beta$ is that of solution $B$. Hence, $\alpha > \beta$ is contradictory to the optimality of $A^{\star}$. Therefore, $\alpha = \beta$, which means  the solution structure $B$ is the optimal solution for problem \eqref{ahmad. raveshe jadide social1}. Therefore, Algorithm 1 will propose the optimal solution of \eqref{ahmad. raveshe jadide social1}. Now consider the case where all   charging stations are not used in full capacity, i.e., the  potential customers are less than available capacity of charging stations. As is it shown in lemma \ref{station.order}, in the optimal solution of problem \eqref{ahmad. raveshe jadide social1}, customers will be assigned to the charging stations starting from charging station $s_1$. The same structure holds in the case where Algorithm \ref{alg.gloabl} adds a virtual station, since the coefficients of decision variables will have the same structure as they have in \eqref{ahmad. raveshe jadide social1}. Therefore, if the available capacity of charging stations is more than the potential set of customers can use, Algorithm \ref{alg.gloabl} will not send any customers to the station $s_{Q+1}$, and in the optimal solution of problem \eqref{ahmad. raveshe jadide social1} all customers will be admitted to the system in full.

\subsection{Proof of Proposition \ref{pm.ic}}

\noindent We know that
\begin{align}&P_{i+1,j,\ell} - P_{i,j,\ell} = v_{i+1}(W_{i,j,\ell} - W_{i+1,j,\ell}) \nonumber\\ &v_{i+1}(W_{i,j,\ell} - W_{i+1,j,\ell}) \geq v_{i}(W_{i,j,\ell} - W_{i+1,j,\ell}),\end{align}
and hence, we can conclude that \begin{align}
    P_{i+1,j,\ell} - P_{i,j,\ell} \geq v_{i}(W_{i,j,\ell} - W_{i+1,j,\ell}),
\end{align} which satisfies the vertical IC constraints using Lemma \eqref{localic lemma}. For proving Horizontal IC, we know from   \eqref{ahmadic2} that
  $ P_{i,j+1,\ell} - P_{i,j,\ell} = v_{i}(W_{i,j,\ell} - W_{i,j+1,\ell})$, which satisfies the condition stated in Lemma \eqref{localic lemma} for Horizontal IC. For proving \eqref{IC4}, we need to show that $P_{i,j,\ell} + v_{i} W_{i,j,\ell} \leq  P_{i,j,m} + v_{i} W_{i,j,m}$ if $m \in \mathcal B_{\ell}$ that we can get with considering constraint \eqref{new.profit.cons} and equations \eqref{ahmadic1}-\eqref{ahmadic3}.
We prove IR by induction for customers with traveling preference $\mathcal{G}_{\ell}$. We know that IR requires that $P_{i,j,\ell} \leq R_i - v_i W_{i,j,\ell}$. Starting with $i=1$ we have $P_{1,j,\ell} = R_1 - v_1 W_{1,j,\ell}$. Now,   assume that  IR holds for type $(i,j,\ell)$. For type $(i+1,j,\ell)$, we can write $P_{i+1,j,\ell} = \big(P_{i,j,\ell} + v_{i+1} W_{i,j,\ell} - v_{i+1}W_{i+1,j,\ell}\big) \leq \big( R_i - v_i W_{i,j,\ell} + v_{i+1}W_{i,j,\ell} - v_{i+1}W_{i+1,j,\ell}\big)$.  Also, we know that $W_{i+1,j,\ell} \leq W_{i,j,\ell} \leq \frac{R_i}{v_i} \leq \big(\frac{R_{i+1} - R_{i}}{v_{i+1} - v_i}\big)$, which leads to 
$P_{i+1,j,\ell} \leq \bigg( R_i + (v_{i+1} - v_{i})\frac{R_{i+1} - R_i}{v_{i+1} - v_i} - v_{i+1}W_{i+1,j,\ell}\bigg)$.
Accordingly,
$P_{i+1,j,\ell} \leq R_i + R_{i+1} - R_i - v_{i+1}W_{i+1,j,\ell} = R_{i+1} - v_{i+1} W_{i+1,j,\ell}$, 
which concludes that: $P_{i+1,j,\ell} \leq R_{i+1} - v_{i+1} W_{i+1,j,\ell} $. This proves IR.

\end{document}